\documentclass{amsart}
\usepackage{amsmath,amssymb,verbatim, graphicx}
\usepackage{latexsym,color}
\usepackage[mathscr]{eucal}

\newtheorem{thm}{Theorem}[section]

\newtheorem{cor}[thm]{Corollary}

\numberwithin{equation}{section}

\newcommand{\up}[1]{\textup{#1}}

\newcommand{\bea}{\begin{eqnarray*}}
\newcommand{\eea}{\end{eqnarray*}}

\newcommand{\while}{\texttt{while}\,}
\newcommand{\Do}{\texttt{do}\,}
\renewcommand{\top}{\mathsf{T}}
\renewcommand{\bot}{\mathsf{F}}
\newcommand{\North}{\textsf{N}}
\newcommand{\South}{\textsf{S}}
\newcommand{\East}{\textsf{E}}
\newcommand{\West}{\textsf{W}}
\newcommand{\Neon}{\textsf{neon}}
\newcommand{\Square}{\textsf{square}}
\newcommand{\comp}{\mathbin{;}}

\newcommand{\fix}{\texttt{fix}}
\newcommand{\Fix}{\texttt{Fix}}

\newcommand{\vor}{\mathbin{\vee\kern -8.5truept \vee}}
\newcommand{\bigor}{\mathop{\bigvee\kern -8.5truept \bigvee}}

\begin{document}

\title{Well structured program equivalence is highly undecidable.}

\author{Robert Goldblatt}
\address{%
School of Mathematics, Statistics and Operations Research\\ Victoria University of Wellington\\ New Zealand}
\email{%
Rob.Goldblatt@msor.vuw.ac.nz}
\author{Marcel Jackson}
\address{%
Department of Mathematics and Statistics\\ La Trobe University\\ Victoria  3086\\
Australia}
\email{%
M.G.Jackson@latrobe.edu.au}
\thanks{The second author was supported by ARC Discovery
Project Grant DP1094578.}
\begin{abstract}
\end{abstract}


\begin{abstract}
We show that strict deterministic propositional dynamic logic with intersection is highly undecidable, solving a problem in the Stanford Encyclopedia of Philosophy.  In fact we show something quite a bit stronger.  We introduce the construction of program equivalence, which returns the value $\top$ precisely when two given programs are equivalent on halting computations.  We show that virtually any variant of propositional dynamic logic has $\Pi_1^1$-hard validity problem if it can express even just the equivalence of well-structured programs with the empty program \texttt{skip}.     We also show, in these cases, that the set of propositional statements valid over finite models is not recursively enumerable, so there is not even an axiomatisation for finitely valid propositions.
\end{abstract}

\maketitle

\section{Introduction}
Determinism has played an unusual role in the study of programs.  While most actual algorithms are deterministic in nature, there has traditionally been a strong theme on modeling programs nondeterministically.  Indeed the standard semantics for classic program logics such as dynamic logic, treat programs as binary relations on the state space of computer, and (in the standard relational semantics) apply constructions such as program union and reflexive transitive closure, which fall outside of conventional programming languages.   Of course, there are numerous good reasons for this: one is attempting to reason about programs more than reason from within them.  Stating that ``property $\alpha$ is true after some number of iterates of $p$'' is a useful assertion to make and close to the kind of questions that need to be asked in applications such as formal program verification.

Another occasionally cited reason for the focus on nondeterminism is that logics based over deterministic programs (partial functions) are known to experience an unexpected explosion in complexity.  In fact this is only half true.  Satisfiability for strict deterministic \textsf{PDL} (deterministic program variables, and program union and $^*$ replaced by only conventional constructions of structured programming: \texttt{if-then-else} and \texttt{while-do}) is only \texttt{PSPACE}-complete \cite{HR}, while the full \textsf{PDL} (over nondeterministic programs), and even strict \textsf{PDL} has \texttt{EXPTIME}-complete complexity (see \cite{HKT} for these and other similar results).  However the introduction of program intersection produces enormous contrast. Standard (that is, nondeterministic) \textsf{PDL} with intersection is decidable \cite{dan}, albeit doubly exponential time complete \cite{lanlut} (a result that has recently been extended to \textsf{PDL} with intersection and converse \cite{gollohlut}) while Harel showed that \emph{deterministic} \textsf{PDL} with intersection (\textsf{DIPDL}) has a $\Pi_1^1$-hard satisfiability problem, at the first level of the analytic hierarchy!

Strangely, it seems unknown what happens between the relatively well behaved \textsf{SDPDL} and the unimaginably badly behaved \textsf{DIPDL}.  The decidability of strict deterministic propositional dynamic logic with intersection (\textsf{SDIPDL}) appears open and indeed is stated as such in the Winter 2008 edition of the Stanford Encyclopedia of Philosophy \cite{balSEP}.  While program intersection is not a conventionally encountered programming construction, it is easy to simulate the intersection of two actual programs $p$ and $q$ and return the result when and if they both halt and agree.  Thus it is an available construct of conventional programming even if it is not expressible within the language of \textsf{SDPDL}.

Recently the second author (with Tim Stokes) has examined algebraic formulations of deterministic program logics and produced a very simple axiomatisation for the loop-free fragment of \textsf{SDIPDL} \cite{jacsto}.  The validity problem of this fragment is easily seen to be \texttt{NP}-complete (by guessing a finite validating model of size polynomial in the complexity of a given formula).  The authors of \cite{jacsto} were rather hopeful that despite Harel's famous negative result for \textsf{DIPDL}, the strict fragment might still be decidable.  In the present article we show this is not the case: \textsf{SDIPDL} also suffers $\Pi_1^1$-hardness.  In fact we show a more general result that concerns variants of \textsf{PDL} that are not necessarily deterministic.  We identify a natural notion of ``program equivalence'' and show that this inevitably leads to $\Pi_1^1$-hardness when expressible in a variant of \textsf{PDL}, independently of the constraint of deterministic atomic programs.  The $\Pi_1^1$-hardness of \textsf{SDIPDL} can be explained by the fact that in deterministic variants of \textsf{PDL}, intersection can be used to express program equivalence.

We also show that for variants of \textsf{PDL} capable of expressing program equivalence (such as \textsf{SDIPDL}) there is no axiomatisation possible for the propositions satisfiable on finite relational models.

\section{Program constructions}
The usual semantics for program intersection is simply set-theoretic intersection of binary relations.  Thus the program $p\cap q$ relates state $s$ to state $t$ provided that both $p$ and $q$ relate $s$ to $t$.  However even if $p\cap q$ relates state $s$ to $t$, enacting $p$ and state $s$ might give rise to some $t'$ outside of the range of the relation $q$.  We consider a reasonable variant of intersection, which we refer to as \emph{program equivalence}.  For programs $p$, $q$, the proposition $p\bowtie q$ (``$p$ tie $q$'', or ``$p$ is equivalent to $q$'') is true at a point $a$ if $p$ is equivalent to $q$ at $a$: in the relational semantics, $p\bowtie q$ has truth set equal to 
\[
\{a\mid (\forall b)\, (a,b)\in p\leftrightarrow (a,b)\in q\}.
\]
Program equivalence can be expressed in \textsf{SDIPDL} as $\langle p\cap q\rangle\top\vee \neg(\langle p\rangle\top\vee \langle q\rangle\top)$.   And, provided query is included, \textsf{SDPDL} with program equivalence can express intersection: $p\cap q=(p\bowtie q)?\comp p$.

Our main results will use a construction weaker than program equivalence.  Consider the unary operation $\Fix$ acting on programs $p$ to produce a proposition $\Fix(p)$ that asserts that halting computations of $p$ act effectlessly.  In the relational semantics, 
\[
\Fix(p)=\{a\mid (\forall b)\, (a,b)\in p\rightarrow a=b\}.
\]
Our main results are expressed in terms of $\Fix$, however in proofs it is more convenient to use a construction $\fix(p)$, which we define as $\Fix(p)\wedge \langle p\rangle\top$.  Note that $\Fix(p)=\fix(p)\vee [p]\bot$, so that $\fix$ and $\Fix$ are interdefinable in any reasonable variant of $\textsf{PDL}$.  But also, $\fix$ (whence $\Fix$) can be expressed in terms of program equivalence as $p\bowtie \texttt{skip}$ (hence it is expressible if $\cap$ is expressible in the deterministic case).  
On the other hand, $\bowtie$ cannot be expressed using $\Fix$ because one can find models of  \textsf{DPDL} that are closed under $\Fix$ but not under program equivalence (we omit the details of this claim).

A key observation in this note is that expressions of the form $[x^*]\alpha$ are expressible in the language of well-structured programs (provided that $x$ and $\alpha$ are): as $[\while \alpha\, \Do x]\bot$.    Expressions of the form $[(x\cup y)^*]\alpha$ 
are fundamental to Harel's original proof of the high undecidability of \textsf{DIPDL}: they are used to interpret an infinite grid.  Expressions of this form are not in general expressible in strict forms of \textsf{PDL}, however the presence of \texttt{fix} enables something similar to be done in enough cases to encode tiling problems.

\section{Tilings}
The undecidability results are proved by encoding tiling problems as originally employed by Harel \cite{har}.
A \emph{finite set of square tiles} is a finite set $\mathcal{T}=\{T_0,\dots,T_{k-1}\}$ of ``\emph{tiles}'' endowed with a pair of binary ``edge'' relations $\sim_h$ (horizontal) and $\sim_v$ (vertical).  We interpret $T_i\mathrel{\sim_h}T_j$ to mean that tile $T_i$ can be placed on the left of tile $T_j$ in a horizontal row.  Likewise $T_i\mathrel{\sim_v}T_j$ is interpreted to mean that $T_i$ can be placed beneath $T_j$ in a vertical column.  A natural and very standard geometric restriction is that if $T_i\mathrel{\sim_h}T_j$ and $T_k\mathrel{\sim_h}T_j$ and $T_k\mathrel{\sim_h}T_\ell$, then $T_i\mathrel{\sim_H}T_\ell$ also.   We will not make use of this restriction, though assuming it does not affect the computational complexity of the tiling problems we consider.

Consider the non-negative integer lattice $\omega\times\omega$ endowed with relations $\sim_h$ and $\sim_v$ defined by $(i,j)\mathrel{\sim_h}(i+1,j)$ and $(i,j)\mathrel{\sim_v}(i,j+1)$ for all $i,j\geq 0$ (here of course, lattice is referring to square grids rather than ordered sets).  A \emph{tiling of the positive quadrant of the plane} (henceforth, a \emph{tiling of the plane}) is a function from $\omega\times\omega$ into $\mathcal{T}$ that preserves the relations $\sim_h$ and $\sim_v$.  Tilings of $\mathbb{Z}\times\mathbb{Z}$ are defined analogously.

We use two fundamental facts on tiling the plane.
\begin{itemize}
\item {\bf Tiling Fact 1.} The following problem is $\Sigma_1^1$-complete.  Given a finite set of tiles $\mathcal{T}$ with distinguished subset $\mathcal{N}$ of ``neon'' tiles.  Is there is a tiling of the plane $\tau$ in which $\tau(0,0)=T_0$ and that $\tau^{-1}(\mathcal{N})\cap \{(i,i)\mid i\in\omega\}$ is infinite (that is, the diagonal contains infinitely many neon tiles).

\item {\bf Tiling Fact 2.} Let $S_{\rm period}$ denote the set of finite sets of square tiles that can tile $\mathbb{Z}\times \mathbb{Z}$ periodically, and let $S_{\rm no tiling}$ denote the set of finite sets of tiles that cannot tile the plane at all.  Then  $S_{\rm period}$ is recursively inseparable from $S_{\rm no tiling}$.
\end{itemize}

Tiling Fact 2 can be found in B\"oger, Gr\"adel and Gurevich \cite[Theorem 3.1.7]{BGG}: tiling periodically means that there is a tiling of $\mathbb{Z}_n\times\mathbb{Z}_m$, with the obvious toroidal adjacency constraints (work modulo $n$ horizontally and modulo $m$ vertically).  Tiling Fact 1 is a minor variant of some well known tiling problems investigated by Harel; see \cite{har2} or \cite{HKT} for example.  We now give a brief sketch of a proof of the $\Sigma_1^1$-completeness claim.  In \cite[p.~233]{har2}, Harel shows that the following problem is $\Sigma_1^1$-complete: given a nondeterministic Turing machine program $T$, with initial state $q_0$ and started on a one-way infinite blank tape, does $T$ return to the state $q_0$ infinitely often? We now reduce this problem to the problem in Tiling Fact 1.   We use a modification of the standard translation of Turing machines into tiles, as presented, say, by Robinson~\cite{rob}.  Using the nomenclature of Robinson's article, there are essentially four kinds of tile (aside from the blank tile which we will not need, as we're only tiling the positive quadrant): the initial tiles (including one designated start tile $T_0$), the merge tiles, the action tiles and the alphabet tiles.  The action tiles are constructed according to the commands of the Turing machine program.  Provided that $T_0$ is placed at the position $(0,0)$, the tiling can only be completed to the $n$th row if the program can run for $n$ steps of computation without halting.  Moreover, each successfully tiled row encodes the configuration of the Turing machine tape at the corresponding step of computation.

Now duplicate all tiles except initial tiles and action tiles.  For each duplicated tile, we make the second copy ``neon'', and adjust the horizontal edge constraints to ensure that neon tiles can be placed horizontally adjacent only to other neon tiles (and even then, only if they additionally satisfy the original edge constraints).  Vertical constraints are unchanged however.  Now, replace every action tile that encodes a transition into the state $q_0$, by a neon copy.  These tiles are not to be duplicated: they are only neon.  Also, action tiles not involving a transition into $q_0$ are never neon.  Then, in any tiling of the plane, a row containing a neon tile must contain only neon tiles.  Since each successfully tiled row can contain precisely one action tile, the following are equivalent: there is a computation that revisits state $q_0$ infinitely often; there is a tiling of the plane starting from $T_0$ and in which infinitely many rows are neon; there is a tiling of the plane starting from $T_0$ and in which infinitely neon tiles are placed on the diagonal.  As the first of these is $\Sigma_1^1$-complete, so the problem in Tiling Fact 1 is $\Sigma_1^1$-hard.  Completeness follows in the usual way.

\section{Main argument}
Let $\mathcal{T}=\{T_0,\dots,T_{k-1}\}$ be some fixed finite set of tiles.  For $i=0,1,\dots,k-1$ we let $\alpha_i$ denote an atomic proposition variable which we think of as corresponding to the placement of tile $T_i$.  
In order to produce our $\omega\times\omega$ grid we introduce four atomic program variables: $\East$, $\West$, $\South$ and $\North$.  Squares of the grid will be created by asserting statements of the form $\texttt{fix}(\North\comp\East\comp\South\comp\West)$.  We first define the propositions required, then explain how these force a tiling.

{\bf Step 0.} Defining a square.    We need to be able to find squares in both clockwise and anti-clockwise directions.  We encode the clockwise square by the following proposition:
\[
\fix(\North\comp\South)\wedge [\North]\fix(\East\comp\West)\wedge [\North\comp \East]\fix(\South\comp\North)\wedge
[\North\comp \East\comp \South]\fix(\West\comp\East)\wedge \fix(\North\comp\East\comp\South\comp\West).
\]
The anticlockwise square is defined in the dual way, following partial paths through $\East\comp\North\comp\West\comp\South$.  We denote the conjunction of the two square propositions by \Square.

{\bf Step 1.} To define a grid we use the statement  $\rho_1$:
\[
[\North^*][\East^*]\Square
\]
which, as observed above, can be expressed using only modal operators and the language of well-structured programs (instead of $^*$).

{\bf Step 2.} To force a tiling, we first let $\alpha$ denote the proposition that asserts that precisely one of the $\alpha_i$ is true.  Then, for each $i$, let $\beta_i$ denote the disjunction of all the atomic tile propositions $\alpha_j$ for which $T_i\sim_hT_j$.  Similarly, we let $\beta^i$ denote the disjunction of the atomic tile propositions $\alpha_j$ for which $T_i\sim_vT_j$.  Then, provided we have an $\omega\times \omega$ grid, a tiling can be forced by $\rho_2$:
\[
[\North^*][\East^*] \left(\alpha\wedge \bigwedge_{i=0}^{k-1}\left(
\alpha_i\Rightarrow ([\East]\beta_i \wedge [\North]\beta^j)\right)\right)
\]

{\bf Step 3.} To force infinitely many neon tiles in the diagonal, first let \Neon\ denote the disjunction of the atomic neon tile propositions.  Then we use $\rho_3$:
\[
[(\North\comp\East)^*]\langle(\North\comp\East)^*\rangle\Neon.
\]

\begin{thm}\label{thm1}
Fix any variation \textsf{VPDL} of \textsf{PDL} capable of expressing the usual connectives on propositions, program composition, \texttt{while}-\texttt{do}, modal operators and $\fix$.  The validity problem for \textsf{VPDL} is $\Pi_1^1$-hard, regardless of whether atomic programs are assumed to be deterministic or not.
\end{thm}
\begin{proof}
For any set of tiles $\mathcal{T}$, with neon subset $\mathcal{N}$, let $\gamma$ denote $\alpha_0\wedge \rho_1\wedge\rho_2\wedge\rho_3$.  We claim that the following are equivalent:
\begin{enumerate}
\item $\mathcal{T}$ can tile the positive quadrant of the plane with infinitely many neon tiles on the diagonal and with $T_0$ in the $(0,0)$ position; 
\item $\gamma$ can be satisfied in some relational model where all atomic programs are deterministic (even injective partial functions);
\item $\gamma$ can be satisfied in some relational model.
\end{enumerate}
Implication $1\Rightarrow2$ is routine, while $2\Rightarrow 3$ is trivial.  Now assume that $\gamma$ is satisfied at some point of a relational model.  We label this point by $a_{0,0}$.  Now by $\rho_1$ we have that $\Square$ holds at $a_{0,0}$.  Thus, the program $\North\comp\East$ is defined at $a_{0,0}$, because $a_{0,0}$ is fixed by $\North\comp\East\comp\South\comp\West$.  Then by $\rho_3$, there is a nontrivial iterate of $\North\comp\East$ at which $\Neon$ is true.  Thus there is a path of edges from $a_{0,0}$ alternating $\North$ and $\East$ and leading to a position at which $\Neon$ is true.   We label the points visited along this path (after $a_{0,0}$) by $a_{0,1}$, $a_{1,1}$, $a_{1,2}$, $a_{2,2},\dots$; see the left picture in Figure \ref{fig1}.  We do not rule out the possibility that some points in the model are labelled more than once: to produce the tiling, we consider only the labels of the selected points

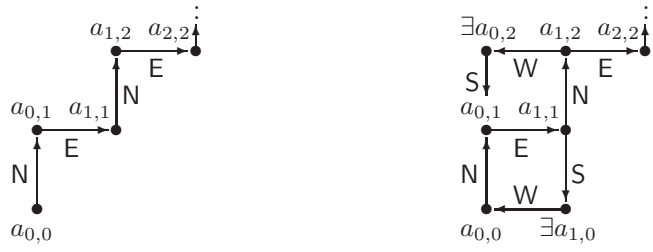
\begin{figure}
\begin{picture}(300,110)
\put(10,10){
\begin{picture}(120,100)
\put(10,10){\circle*{4}}\put(0,0){$a_{0,0}$}
\put(10,40){\circle*{4}}\put(0,46){$a_{0,1}$}
\put(40,40){\circle*{4}}\put(22,46){$a_{1,1}$}
\put(70,70){\circle*{4}}\put(52,76){$a_{2,2}$}
\put(40,70){\circle*{4}}\put(30,76){$a_{1,2}$}
\put(10,10){\vector(0,1){27}}\put(0,20){$\North$}
\put(10,40){\vector(1,0){27}}\put(20,30){$\East$}
\put(40,70){\vector(1,0){27}}\put(42,50){$\North$}
\put(40,40){\vector(0,1){27}}\put(52,60){$\East$}
\put(70,70){\vector(0,1){10}}\put(69,82){$\vdots$}
\end{picture}
}
\put(180,10){
\begin{picture}(120,100)

\put(10,10){\circle*{4}}\put(0,0){$a_{0,0}$}
\put(10,40){\circle*{4}}\put(0,46){$a_{0,1}$}
\put(40,40){\circle*{4}}\put(22,46){$a_{1,1}$}
\put(70,70){\circle*{4}}\put(52,76){$a_{2,2}$}
\put(40,70){\circle*{4}}\put(30,76){$a_{1,2}$}
\put(10,10){\vector(0,1){27}}\put(0,20){$\North$}
\put(10,40){\vector(1,0){27}}\put(20,30){$\East$}
\put(40,70){\vector(1,0){27}}\put(42,50){$\North$}
\put(40,40){\vector(0,1){27}}\put(52,60){$\East$}
\put(70,70){\vector(0,1){10}}\put(69,82){$\vdots$}
\put(40,10){\circle*{4}}\put(30,0){$\exists a_{1,0}$}
\put(40,10){\vector(-1,0){27}}\put(20,12){$\West$}
\put(40,40){\vector(0,-1){27}}\put(42,20){$\South$}

\put(10,70){\circle*{4}}\put(0,76){$\exists a_{0,2}$}
\put(40,70){\vector(-1,0){27}}\put(20,60){$\West$}
\put(10,70){\vector(0,-1){17}}\put(2,55){$\South$}

\end{picture}
}
\end{picture}
\caption{Selecting the points $a_{i,j}$, and completing the $\omega\times\omega$ grid.}\label{fig1}
\end{figure}

Now as $\Square$ holds at $a_{0,0}$, we have that after $\North\comp\East$ it is necessary that $\fix(\South\comp \North)$ hold.  Hence, in particular there is a point $a_{1,0}$ that is reached by an application of $\South$ from the point $a_{1,1}$.  Again applying $\Square$ at $a_{0,0}$, we have that after applying $\North\comp\East\comp \South$ it is necessary that $\fix(\West\comp \East)$.  Thus in particular, there is a point $a_{0,0}'$ west of $a_{1,0}$.  However $a_{0,0}'$ is reached by an application of $\North\comp\East\comp \South\comp \West$, which by $\Square$ must fix $a_{0,0}$.  Hence $a_{0,0}'=a_{0,0}$.  

Similarly, $\rho_1$ ensures that $\Square$ is true at the point $a_{0,1}$.  We now construct a square anticlockwise through points $a_{0,1}$, $a_{1,1}$, $a_{1,2}$ and some new point $a_{0,2}$.  The idea is essentially dual to the previous case: after applying $\East;\North$ (reaching $a_{1,2}$), it is necessary that $\fix(\West\comp\East)$ be defined, thus we encounter some new point $a_{0,2}$.  From here a further $\South$ is forced, and then as $\East\comp\North\comp\West\comp\South$ fixes $a_{0,1}$, we have the desired square.  

So far we have not used all the power of the proposition $\Square$: in the right hand picture in Figure \ref{fig1}, the bottom left square has a different orientation to the square above it.  However, each time we extended a new arrow from a point, we did so by way of propositions of the form $\fix(\East\comp\West)$ (and so on): thus in fact every arrow drawn has an associated converse arrow labelled with the appropriate dual name ($\East$ switched with $\West$ and $\North$ switched with $\South$).  Once these edges are also drawn, both squares so far obtained are identical (two-way edges, with dual labels).  So in fact, the process can be continued, working out outward from the central diagonal (with clockwise constructions below the horizontal and anti-clockwise constructions above)
until a rectangular grid has been formed.  

Then we apply $\rho_3$ a further time: extending the diagonal to a new point $a_{n,n}$ where $\Neon$ is defined, and filling out the remaining pieces of a larger rectangle and so on.

In this way an infinite grid is interpreted, with neon tile propositions holding at infinitely many places on the diagonal.  Furthermore, every position in this grid can now be visited by first iterating $\East$ and then iterating $\North$.  Now $\gamma$ forces $\alpha_{0}$ to be true at $a_{0,0}$.  And then, working inductively outward from $a_{0,0}$, the proposition $\rho_2$ ensures that a tiling proposition holds at every one of the selected points and that neighbouring squares (horizontally or vertically) have tiling propositions that match the tiling constraints. Thus we interpreted  a tiling of the positive quadrant of the plane in which neon tiles occur infinitely often along the diagonal.  As the problem in Tiling Fact 1 is $\Sigma_1^1$-complete, thus satisfiability for \textsf{VPDL} is $\Sigma_1^1$-hard and validity is $\Pi_1^1$-hard.
\end{proof}
Recall that if atomic programs are deterministic, then intersection can be used to define $\fix$ on well-structured programs.  This gives the following corollary.
\begin{cor}
Satisfiability for \textsf{SDIPDL} is $\Pi_1^1$-hard.
\end{cor}
Consider the operation of \emph{program difference}: 
\[
p-q:=\{(a,b)\mid (a,b)\in p\mbox{ and} (a,b)\notin q\}.
\]
It is well known that standard \textsf{PDL} with program \emph{complementation} is undecidable (see \cite[Theorem 10.12]{HKT}).  Program difference can be expressed in terms of program complementation, but the reverse need not be true in the absence of a universal program (that is, the universal relation in the relational semantics).  As a second corollary, we show that standard \textsf{PDL} with program difference is $\Pi_1^1$-hard.
\begin{cor}
\textsf{PDL} with program difference \up(whence with program complementation\up) is $\Pi_1^1$-hard.
\end{cor}
\begin{proof}
First observe that program intersection can be expressed from program difference: $p\cap q=p-(p-q)$.  Now observe that $\fix(p)=(\langle p\rangle \top)\wedge ([p-(p\cap \texttt{skip})]\bot)$.
\end{proof}

\begin{thm}\label{thm2}
Fix any variation \textsf{VPDL} of \textsf{PDL} capable of expressing the usual connectives on propositions, program composition, \texttt{while}-\texttt{do}, modal operators and $\fix$. The set of \textsf{VPDL} propositions valid over finite relational models of \textsf{VPDL} is not recursively enumerable, whence there is no axiomatisation for \textsf{VPDL} over finite models.
\end{thm}
\begin{proof}
Consider a finite set of tiles $\mathcal{T}$, and consider the proposition $\gamma_\mathcal{T}:=\rho_1\wedge \rho_2$.  We first show that if $\gamma_\mathcal{T}$ is satisfied at some point $a_{0,0}$ in a model then $\mathcal{T}$ can tile the plane (whence $\mathcal{T}\notin \mathcal{S}_{\rm notiling}$).  The argument is similar to that used to prove Theorem \ref{thm1}, but we use $\rho_1$ to produce the diagonal (there are no neon tiles to consider).
By $\rho_1$, the proposition $\Square$ is true, which yields points $a_{0,1}$, $a_{1,1}$ and $a_{1,0}$, reached successively in following $\North\comp\East\comp\South$, with $\West$ taking $a_{1,0}$ back to $a_{0,0}$, and with $\East\comp\North\comp\West\comp\South$ following through the points in reverse order.  Now, by $\rho_1$ again, $\Square$ is true at $a_{1,1}$.  Thus we obtain points $a_{1,2}$, $a_{2,2}$ and $a_{2,1}$ forming the rest of a new square based at $a_{1,1}$.  Now we can fill out these points to a $2\times 2$ region using the same argument in the proof of Theorem \ref{thm1}.  Then $\rho_1$ guarantees that $\Square$ is true at $a_{2,2}$ and so on.  Finally, once an $\omega\times\omega$ grid is interpreted, we can use $\rho_2$ to show that precisely one tiling proposition is true at $a_{0,0}$, and then force a tiling as in the proof of Theorem \ref{thm1}.

Now observe that if $\mathcal{T}$ can tile periodically: that is, can tile the torus $\mathbb{Z}_n\times\mathbb{Z}_m$, then $\gamma_\mathcal{T}$ can be satisfied in some finite model based on the $nm$ points of $\mathbb{Z}_n\times\mathbb{Z}_m$.

Thus the set $\mathscr{S}$ of finitely satisfiable propositions contains $\{\gamma_\mathcal{T}\mid \mathcal{T}\in S_{\rm period}\}$ and is disjoint from $\{\gamma_\mathcal{T}\mid \mathcal{T}\in S_{\rm notiling}\}$.  Now $\mathscr{S}$  is recursively enumerable (simply search for a finite satisfying model).  But it cannot be recursive, because $S_{\rm period}$ and $S_{\rm notiling}$ are recursively inseparable.  Hence $\mathscr{S}$ is not \textsf{coRE}.  Whence the propositions valid over finite models of \textsf{VPDL} is not \textsf{RE}.
\end{proof}
We mention that in order to express $\Fix$ in terms of program equivalence we invoked the program \texttt{skip}.  In the absence of \texttt{skip} (whence also query, as $\texttt{skip}=\top?$), it is unclear if Theorem \ref{thm1} and Theorem \ref{thm2} hold (replacing  $\fix$ by program equivalence).  However all of the arguments relating to the encoding of tilings can be routinely adapted to the program equivalence situation, with some simplification.  As a sketch: work with only $\North$ and $\East$, and replace the proposition $\Square$ by statements of the form $(\North\comp\East)\bowtie (\East\comp\North)$.

\medskip

\noindent{\bf Acknowledgement}\\
The authors are indebted to Dr.\ Tim Stokes for initiating the investigation into program equivalence in publications such as \cite{dfssto1,dfssto2,sto} as well as for numerous discussions and feedback during the writing of this article.

\end{document}